\documentclass[twocolumn,aps,pra,reprint,superscriptaddress,amsmath,amssymb,bbm,floatfix]{revtex4-2}

\usepackage[latin9]{inputenc}
\usepackage{color}
\usepackage{amstext}
\usepackage{graphicx}
\usepackage[colorlinks, allcolors={blue}]{hyperref}
\usepackage{physics}
\usepackage{blindtext}
\usepackage{pifont}
\usepackage{xcolor,graphicx}
\usepackage{newlfont}
\usepackage{amssymb,amsmath,mathrsfs,amsthm}
\usepackage{verbatim}
\usepackage{bbm}
\usepackage{multirow}
\usepackage[varg]{txfonts}

\newcommand{\orcid}[1]{\href{https://orcid.org/#1}{\includegraphics[width=7pt]{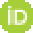}}}
\newcommand{\cmark}{\ding{52}}
\newcommand{\xmark}{\ding{56}}

\newcommand{\be}{\begin{equation}}
\newcommand{\ee}{\end{equation}}
\newcommand{\beq}{\begin{eqnarray}}
\newcommand{\eeq}{\end{eqnarray}}
\newcommand{\HS}{\mathrm{HS}\,}
\newcommand{\He}{\mathrm{He}\,}
\newcommand{\Bu}{\mathrm{Bu}\,}
\newcommand{\bignorm}[1]{\mbox{$ \left|\!\left| #1 \right|\!\right|$}}
\newcommand{\fR}{\mathfrak{R}}
\newcommand{\fI}{\mathfrak{I}}
\newcommand{\tD}{\widetilde{D}}
\newcommand{\mc}[1]{\mathcal{#1}}
\newcommand{\mf}[1]{\mathfrak{#1}}
\newcommand{\mbb}[1]{\mathbbm{#1}}
\newcommand{\bigdbar}{\,\Big|\Big|\,}

\renewcommand\bra[1]{{\langle{#1}|}}
\renewcommand\ket[1]{{|{#1}\rangle}}
\renewcommand{\Tr}{\text{Tr}}

\newtheorem{lemma}{Lemma}

\theoremstyle{remark}
\newtheorem{axiom}{Axiom}

\newtheorem{example}{Example}

\allowdisplaybreaks

\begin{document}

\title{Geometric monotones of violations of quantum realism}

\author{Alexandre C. Orthey Jr. \orcid{0000-0001-8111-3944}}
\email{alexandre.orthey@gmail.com}
\affiliation{Institute of Fundamental Technological Research, Polish Academy of Sciences, Pawi\'nskiego 5B, 02-106 Warsaw, Poland.}

\author{Alexander Streltsov 
\orcid{0000-0002-7742-5731}}
\affiliation{Institute of Fundamental Technological Research, Polish Academy of Sciences, Pawi\'nskiego 5B, 02-106 Warsaw, Poland.}

\begin{abstract}
Quantum realism, as introduced by Bilobran and Angelo [EPL 112, 40005 (2015)], states that projective measurements in quantum systems establish the reality of physical properties,
even in the absence of a revealed outcome. This framework provides a nuanced perspective on the distinction between classical and quantum notions of realism, emphasizing the contextuality and complementarity inherent to quantum systems. While prior works have quantified violations of quantum realism (VQR) using measures based on entropic distances, here we extend the theoretical framework to geometric distances. Building on an informational approach, we derive geometric monotones of VQR using trace distance, Hilbert-Schmidt distance, Schatten $p$-distances, Bures, and Hellinger distances. We identify Bures and Hellinger distances as uniquely satisfying all minimal criteria for a bona fide VQR monotone. Remarkably, these distances can be expressed in terms of symmetric R\'enyi and Sandwiched R\'enyi divergences, aligning geometric and entropic approaches. Our findings suggest that the realism-information relation implies a deep connection between geometric and entropic frameworks, with only those geometric distances expressible as entropic quantities qualifying as valid monotones of VQR. This work highlights the theoretical and practical advantages of geometric distances, particularly in contexts where computational simplicity or symmetry is important.
\end{abstract}

\maketitle

\section{Introduction}

The concept of quantum realism of observables, as proposed by Bilobran and Angelo \cite{bilobran2015measure}, can be summarized as follows. After the projective measurement of a physical property in a quantum system, we can assert that the property is now well-defined, i.e., it is real for the state of the system. Any subsequent projective measurement of the same property will yield the same outcome because the new measurement merely reveals a pre-existing \textit{state of affairs}---or reality. This reasoning also holds for non-selective (or non-revealed) projective measurements. Even if the experimentalist is not present in the lab when the measurement occurs or forgets the result, the measurement still takes place, establishing the property as real. Thus, a projective measurement, even without a revealed outcome, serves as the mechanism by which realism emerges in a quantum system. 

Consider, for instance, the measurement of the observable $\sigma_z = \ket{0}\bra{0} - \ket{1}\bra{1}$ on the state $\ket{\psi} = (\ket{0} + \ket{1}) / \sqrt{2}$. After the measurement, the system's state will collapse to either $\ket{0}$ or $\ket{1}$. At this point, the realism (or reality) of $\sigma_z$ is established, meaning that any subsequent measurement of $\sigma_z$ will yield the same outcome. In the case of a non-selective measurement, however, the post-measurement state of the system would be the mixed state $\ket{0}\bra{0} + \ket{1}\bra{1}$. This mixed state indicates that the system is definitely in either $\ket{0}$ or $\ket{1}$, but the observer lacks knowledge of which one. This lack of knowledge is subjective and classical---it reflects the observer's ignorance rather than an intrinsic uncertainty in the quantum state, as is the case with superpositions. By contrast, the observable $\sigma_x$ is real for the initial state $\ket{\psi}$ because the state is an eigenstate of $\sigma_x$, with a well-defined eigenvalue. This highlights the fundamental difference between realism in quantum mechanics and classical realism. Classical realism assumes that every physical property of a system is well-defined at all times, independent of observation or the presence of an observer. In contrast, quantum realism recognizes that certain pairs of incompatible observables, such as $\sigma_x$ and $\sigma_z$, cannot always be simultaneously well-defined. This limitation is intrinsic to the nature of quantum systems and reflects the contextuality and complementarity that are hallmarks of quantum theory.

The situation changes dramatically when we consider bipartite systems. If Alice and Bob share a maximally entangled pair of spin-$1/2$ particles, their local states are maximally mixed. Consequently, the spin observable for, say, Alice's particle is \textit{locally} an element of physical reality, as a non-selective measurement of the identity operator does not alter the state. However, when considering the \textit{global} description of the pair, any local non-selective measurement perturbs the global state. This demonstrates that the state of the pair violates quantum realism for the observable associated with either particle. This phenomenon diverges sharply from the Einstein-Podolsky-Rosen (EPR) notion of elements of physical reality \cite{einstein1935can}, as applied to bipartite entangled states. According to EPR, local measurements should not influence the state of affairs at remote locations.

A method to quantify violations of quantum realism (VQR) was proposed in \cite{bilobran2015measure}, where the authors referred to VQR as \textit{irrealism}. They introduced an entropic distance between the measured state and the original state as a means of quantification. These measures of VQR were largely explored in theoretical \cite{dieguez2018information,Rudnicki2018uncertainty,orthey2019nonlocality,engelbert2020hardy,costa2020information,Lustosa2020irrealism,Basso2021complete,Paiva2023coherence,Engelbert2023considerations,Caetano2024quantum,fucci2024theory,orthey2025high,Fucci2025emergence} and experimental \cite{mancino2018information,Dieguez2022Apr,Chrysosthemos2023quantum,Basso2022reality,Araujo2024quantum} approaches. Building on the informational framework of \cite{dieguez2018information} and the experimental findings of \cite{mancino2018information, Dieguez2022Apr}, \cite{orthey2022quantum} presented a set of physically motivated axioms for monotones and measures of VQR, where the later requires more conditions than the former. In their work, they define a notion of conditional quantum information---the complement of quantum conditional entropy---using general relative entropies, also called divergences. They further show how variations in conditional information can be used to derive expressions for VQR measures. Their argument relies on a \textit{realism-information relation}: since non-selective projective measurements are responsible for the emergence of realism, the increase in conditional information about the system, as encoded in the environment, could serve as a monotone of VQR. This connection is formalized using the Stinespring dilation theorem \cite{nielsen2010quantum}.

\subsection{Contributions}

In this work, we extended the concept of conditional information based on relative entropies from \cite{orthey2022quantum} to geometric distances and successfully derived geometric monotones of VQR. Geometric distances, or simply \textit{distances}, share properties with relative entropies, such as contraction under completely positive trace-preserving (CPTP) maps, invariance under unitary transformations, and joint convexity. Additionally, distances are inherently symmetric and satisfy the triangle inequality, making them advantageous in certain contexts. Moreover, some distances are computationally simpler to evaluate, offering a practical benefit over entropic measures. Although none of the distances we studied are additive, making them not good candidates for measures of VQR, some of them can still be classified as monotones of VQR. Among the distances we explore, we have trace distance, Hilbert-Schmidt distance \cite{nielsen2010quantum}, Schatten $p$-distances \cite{bhatia2013matrix,spehner2017geometric}, and Bures \cite{bures1969extension,spehner2017geometric} and Hellinger \cite{roga2016geometric,spehner2017geometric} distances. We found out that, among these distances, Bures and Hellinger distances are the only candidates that satisfy all the minimal requirements for a bona-fide monotone of VQR. Trace distance, for instance, provides a quantifier that fails to identify VQR for some states. The Hilbert-Schmidt distance, on the other way, provides a quantifier of VQR that decreases with the addition of completely uncorrelated states. Interestingly, Bures and Hellinger distances can be written as functions of R\'enyi and Sandwiched R\'enyi divergences in the particular cases where these divergences are symmetric. Our results suggest that, if one assumes the realism-information relation, then only the geometric distances that can be framed as entropic quantities are suitable to derive monotones of VQR.

\section{Fundamental concepts}

\subsection{Violations of quantum realism}

Let us mathematically formalize the concept of quantum realism presented in the introduction by following Refs. \cite{bilobran2015measure,dieguez2018information,orthey2022quantum}. Let $A = \sum_a a A_a$ be a $d$-output observable acting on $\mathcal{H}_\mathcal{A}$, where $A_a = \ket{a}\bra{a}$ are projectors satisfying $\sum_a A_a = \mbb{1}_\mathcal{A}$. Additionally, define the non-selective projective measurement of the observable $A$ as $\Phi_A(\rho_\mc{AB}) = \sum_a \left( A_a \otimes \mbb{1}_\mathcal{B} \right) \rho_\mc{AB} \left( A_a \otimes \mbb{1}_\mathcal{B} \right)$, where $\rho_\mc{AB}$ is the state of a bipartite system in $\mathcal{H}_\mathcal{A}\otimes\mc{H}_\mc{B}$. We say that $\rho_{\mathcal{AB}}$ has realism defined for $A$ \textit{iff} $\Phi_A(\rho_{\mathcal{AB}}) = \rho_{\mathcal{AB}}$. States that already possess realism for a particular observable are invariant under a non-selective (or non-revealed) projective measurement of that observable, i.e., $\Phi_A(\Phi_A(\rho)) = \Phi_A(\rho)$. Henceforth, \textit{the violation of the quantum realism of $A$ given $\rho$} is defined as
\begin{equation}\label{irrealism}
\mathfrak{I}_A(\rho) \coloneqq S\big(\Phi_A(\rho)\big) - S(\rho),
\end{equation}
which can also be called the \textit{irrealism} of $A$ given $\rho$, and $ S(\rho) \coloneqq -\Tr\big(\rho \log_2 \rho\big)$ is the von Neumann entropy. Moreover, we can decompose irrealism as the sum of local coherence and non-optimized quantum discord, that is, $\mathfrak{I}_A(\rho) = \mathcal{C}(\rho_\mathcal{A}) + D_A(\rho)$, where $\mathcal{C}(\rho_\mathcal{A}) = \mathfrak{I}_A(\text{Tr}_\mathcal{B} \rho)$, $D_A(\rho) = I(\mathcal{A}:\mathcal{B})_\rho - I(\mathcal{A}:\mathcal{B})_{\Phi_A(\rho)}$, and $I(\mathcal{A}:\mathcal{B})_\rho = S(\rho_\mathcal{A}) + S(\rho_\mathcal{B}) - S(\rho_\mathcal{AB})$ is the quantum mutual information. Consequently, correlations between parties $\mathcal{A}$ and $\mathcal{B}$ prevent the existence of elements of physical reality for observables on both parties. That is, any positive value of $D_{A(B)}(\rho)$ implies non-zero irrealism $\mathfrak{I}_{A(B)}(\rho)$.

Functionals such as \eqref{irrealism} were also introduced in other works but with different names and different meanings, such as measures of wave-like information \cite{angelo2015wave} and quantum incoherent relative entropy \cite{chitambar2016assisted}.

\subsection{Quantifying violations of realism using quantum conditional information}\label{sec_realism_conditional}

Because the shape of \eqref{irrealism} was chosen ad-hoc, one could wonder if other entropic distances would be a better fit to quantify VQR. To answer that question, an axiomatic and physically motivated approach was taken in Ref. \cite{orthey2022quantum} to propose that irrealism should decrease in the presence of a flux of information from the system to the environment. Since measurements are the key elements by which a physical property gets defined, one could argue that quantifying how much information the environment $\mc{E}$ gets about the system $\mc{S}$ makes it possible to quantify how irrealism decreases. The figure of merit in this case is the complement of quantum conditional entropy
\be\label{cond_entropy}
H(\mc{E}|\mc{S})_\Omega\coloneqq S(\Omega)-S(\Omega_\mc{E}),
\ee
that is, the quantum conditional information \cite{orthey2022quantum}
\be\label{cond_information}
I(\mc{E}|\mc{S})_\Omega\coloneqq \ln d_\mc{E}-H(\mc{E}|\mc{S})_\Omega,
\ee
where $\Omega$ acts over $\mc{H_S}\otimes\mc{H_E}$, $\Omega_\mc{E}=\Tr_\mc{S}\Omega$ and $d_\mc{E}=\dim\mc{H_E}$. While the former measures the ignorance that we still have about $\mc{E}$ given that we know $\mc{S}$, the latter measures the knowledge (or information) we can obtain by accessing $\mc{E}$ conditioned to the information codified in $\mc{S}$. Note that $0\leqslant I(\mc{E}|\mc{S})_\rho\leqslant \ln d_\mc{E}d_\mc{S}$.

To connect conditional information and quantum realism, the authors in \cite{orthey2022quantum} consider a system in the state $\rho_\mc{S}$ that interacts with an environment in the state $\ket{0}\bra{0}$ such that $\Omega_0 \coloneqq \rho_\mc{S}\otimes\ket{0}\bra{0}$ and $\Omega_t \coloneqq U_\mc{SE} \Omega_0 U_\mc{SE}^\dagger$ are the global states before and after the interaction $U_\mc{SE}$. This interaction must result in a non-selective projective measurement over the system after we trace out the environment, that is,
\be\label{traceE}
\Tr_\mc{E}\left[ U_\mc{SE}\left(\rho_\mc{S}\otimes\ket{0}\bra{0}\right)U_\mc{SE}^\dagger\right]= \Phi_A(\rho_\mc{S}),
\ee
where $A$ is some observable of interest (see previous section). Equation \eqref{traceE} guarantees that the information about the observable $A$ in $\mc{S}$ that is going to be codified in $\mc{E}$---through the unitary interaction $U_\mc{SE}$---is associated with the emergence of realism about $A$ in the system---through the map $\Phi_A$. The authors defined that this increment of realism about $A$ given $\rho$, denoted by $\fR_A(\rho)$, should vary such that
\be\label{realism_inf_relation}
\Delta \fR_A(t) \equiv \Delta I(\mathcal{E|S})_{\Omega_0\to\Omega_t}
\ee
which can be thought of as a \textit{realism-information relation}, where $\Delta \fR_A(t) := \fR_A(\rho_\mc{S}(t)) - \fR_A(\rho_\mc{S}(0))$, with $\rho_\mc{S}(0)\coloneqq\rho_\mc{S}$ and $\rho_\mc{S}(t)\coloneqq\Tr_\mc{E}(\Omega_t)=\Phi_A(\rho_\mc{S})$. Additionally, the corresponding change in the environmental conditional information is $\Delta I(\mathcal{E|S})_{\Omega_0\to\Omega_t} := I(\mathcal{E|S})_{\Omega_t} - I(\mathcal{E|S})_{\Omega_0}$. In \cite{orthey2022quantum}, the authors prove that relation \eqref{realism_inf_relation} results in
\be\label{realism_vonNeumann}
\fR_A(\rho_\mc{S})=\ln d_A - S(\rho_\mc{S}||\Phi_A(\rho_\mc{S}))=\ln d_A - \fI_A(\rho_\mc{S}),
\ee
since the von Neumann relative entropy \cite{nielsen2010quantum}
\be\label{relative_entropy}
S(\rho||\sigma)\coloneqq \Tr(\rho\ln\rho-\rho\ln\sigma)
\ee
satisfies $S(\rho_\mc{S}||\Phi_A(\rho_\mc{S}))=S(\Phi_A(\rho_\mc{S}))-S(\rho_\mc{S})$ \cite{costa2020information}. Equation \eqref{realism_vonNeumann} shows that, by taking an informational approach to the problem of the emergence of realism through non-selective projective measurements, we can regain the expression for irrealism of observables \eqref{irrealism}.

In addition to the realism-information relation \eqref{realism_inf_relation}, the authors in \cite{orthey2022quantum} also introduced a list of axioms that should be satisfied by any bona-fide \textit{monotone} of VQR:

\begin{axiom}[Realism and Measurements]\label{A_reality_meas}
The realism of $A$ acting on $\mc{H_A}$ given $\rho$ acting on $\mc{H_A}\otimes\mc{H_B}$, denoted by $\fR_A(\rho)$, is a non-negative real number and increases upon non-selective projective measurements of $A$, reaching a maximum when $\rho=\Phi_A(\rho)$, that is $0 \leq \fR_A(\rho) \leq \fR_A(\mc{M}_A^\epsilon(\rho)) \leq \fR_A(\Phi_A(\rho)) = \fR_A^{\max}$, where $\mc{M}_A^\epsilon(\rho)\coloneqq (1-\epsilon)\rho+\epsilon\Phi_A(\rho)$ is the monitoring map.
\end{axiom}

\begin{axiom}[Role of Other Parts]\label{A_role_parts}
(a) Discarding subsystems does not reduce realism, i.e. $\fR_A(\Tr_{\mathcal{X}}(\rho)) \geq \fR_A(\rho)$ for $\mc{H_X}\subseteq\mc{H_B}$, and (b) adding an uncorrelated system does not change realism, i.e. $\fR_A(\rho \otimes \sigma) = \fR_A(\rho)$.
\end{axiom}

\begin{axiom}[Uncertainty Relation]\label{A_uncertainty}
Two observables $X$ and $Y$ acting on $\mc{H_A}$ cannot be simultaneous elements of reality in general, that is, $\fR_X(\rho) + \fR_Y(\rho) \leq 2\fR_A^{\max}$. Equality must hold only when $[X,Y]=0$ or $\rho=\mbb{1}_\mc{A}/d_\mc{A}\otimes\rho_\mc{B}$.
\end{axiom}

\begin{axiom}[Mixing]\label{A_mix}
For a probabilistic ensemble $\{p_i, \rho_i\}$, realism is concave, i.e. $\fR_A\left(\sum_i p_i \rho_i\right) \geq \sum_i p_i \fR_A(\rho_i)$.
\end{axiom}

Moreover, if a quantifier of VQR additionally satisfies additivity and flag additivity, it is called a \textit{measure} of VQR. Since none of the distances we studied are additive, only monotones of VQR were found.

\subsection{Alternative ways of quantifying quantum conditional information}\label{sec_inf_divergences}

An alternative way of expressing quantum conditional entropy is by means of divergences \cite{tomamichel2014relating}, being von Neumann relative entropy as the most prominent example. From \eqref{cond_entropy} and \eqref{relative_entropy} it can be shown that
\be\label{cond_entropy_relative}
H(\mc{E|S})_\Omega = -S(\Omega||\Omega_\mc{S}\otimes\mbb{1}_\mc{E}).
\ee
On the other way, the quantum conditional information \eqref{cond_information} can be conveniently rewritten in terms of \eqref{relative_entropy} as
\be\label{cond_information_relative}
I(\mc{E|S})_\Omega = S\left(\Omega\bigdbar \Omega_\mc{S}\otimes\frac{\mbb{1}_\mc{E}}{d_\mc{E}} \right).
\ee
To better understand the meaning of conditional information, one can check that it can be decomposed as
\be\label{inf_cond_decomposition}
I(\mc{E|S})_\Omega = I(\Omega_\mc{E})+I(\mc{E:S})_\Omega,
\ee
where $I(\rho_\mc{E})\coloneqq S\left(\Omega_\mc{E}||\mbb{1}_\mc{E}/d_\mc{E}\right) = \ln d_\mc{E}-S(\Omega_\mc{E})$ is the informational content of part $\mc{E}$ (because it is the complement of the entropy of $\mc{E}$) and $I(\mc{E:S})_\Omega=S\left(\rho||\rho_\mc{S}\otimes\rho_\mc{E}\right)$ is the quantum mutual information now written as a divergence. Since the mutual information is a measure of total correlations between the parts, the conditional information can be said to be composed of ``local'' and ``global'' information \cite{orthey2017asymptotic}. In the case of a product state such as $\rho_\mc{S}\otimes\sigma_\mc{E}$, the conditional information will be codified only in part $\mc{E}$. However, if the global state is maximally entangled, all the conditional information will be due to the correlations, since the local state assigned to part $\mc{E}$ is the maximally mixed one.

In \cite{tomamichel2014relating}, it was observed and later expanded upon in \cite{orthey2022quantum} within the context of irrealism measures that if one intends to explore alternative formulations of conditional entropy using other entropic measures, such as R\'{e}nyi and Tsallis entropies, the approach outlined in \eqref{cond_entropy} proves insufficient. This insufficiency arises due to the failure of data-processing inequalities. The same does not apply to approaches using divergences such as the one employed in \eqref{cond_entropy_relative}. By replacing the von Neumann relative entropy $S(\cdot||\cdot)$ by other divergence measures such as sandwiched R\'{e}nyi divergence \cite{tomamichel2014relating} and Tsallis \cite{tsallis1988possible} relative entropy, we can construct alternative entropic measures for conditional information that still attain to \eqref{inf_cond_decomposition}. Incidentally, these other measures for conditional information result in alternative ways to measure the quantum realism of observables through the realism-information relation \eqref{realism_inf_relation} as extensively explored in \cite{orthey2022quantum}.

\subsection{Geometric distances and norms}\label{Sec_norms}

In this paper, we explore the possibility of using geometrical distances to quantify conditional information---and consequently quantum realism---as an alternative to entropic measures. In what follows, we revisit the definition of \textit{distance} and list the examples that we are interested in exploring, as well as their fundamental properties \cite{kreyszig1978introductory}.

A \textit{distance} is a function $d:M\times M\mapsto \mathbb{R}$ that maps each par $x,y\in M$ onto a real number $d(x,y)$ satisfying: \textit{positive definiteness} $d(x,y)\geqslant 0$, where the equality holds \textit{iff} $x=y$; \textit{symmetry} $d(x,y)=d(y,x)$; and \textit{triangle inequality} $d(x,z)\leqslant d(x,y)+d(y,z)$. Positive definiteness is essential to obtain proper distinguishing measures. In addition, other properties are necessary for such measures to have a physical meaning. In the context of density states $\rho$ and $\sigma$, a distance is said to be \textit{contractive} under a CPTP map $\Lambda$ \textit{iff}
\be\label{contractivity}
d(\Lambda(\rho),\Lambda(\sigma))\leqslant d(\rho,\sigma).
\ee
In particular, contractive distances are unitarily invariant, that is, $d(U\rho U^\dagger,U\sigma U^\dagger)= d(\rho,\sigma)$.

We can use norms to find proper distances. One norm of particular interest is the Schatten $p$-norm \cite{bhatia2013matrix} of an operator $X$, given by $\norm{X}_p\coloneqq\left[\Tr\left(|X|^p\right) \right]^{\frac{1}{p}}$, where $p$ is a real number such that $p\geqslant1$ and $|X|\coloneqq\sqrt{X^\dagger X}$. The Schatten $p$-norms of density operators $\rho$ (i.e. $\rho\geqslant 0$, $\Tr\rho=1$, and $\rho^\dagger=\rho$) satisfy  \cite{bhatia2013matrix}: unitary invariance, i.e. $\norm{U\rho V^\dagger}_p=\norm{\rho}_p$ for isometries $U$ and $V$; and multiplicativity under tensor products, i.e. $\norm{\rho\otimes\sigma}_p=\norm{\rho}_p\norm{\sigma}_p$.

Without further ado, the \textit{$L_p$-distances} $d_p$ (also called \textit{Schatten $p$-distances}) \cite{bhatia2013matrix,spehner2017geometric}
\be\label{Lp_distance}
d_p(\rho,\sigma)\coloneqq\norm{\sigma-\rho}_p=\left(\Tr|\sigma-\rho|^p \right)^\frac{1}{p},
\ee
for all $p\geqslant 1$. For $p=1$ and $p=2$, the $L_p$-distances are called \textit{trace distance} $d_{\Tr}$ and \textit{Hilbert-Schmidt distance} $d_\HS$, respectively. Although all the $L_p$-distances are unitarily invariant because they inherit this property from the Schatten $p$-norms, the only $L_p$-distance that is contractive under CPTP maps is the trace distance \cite{perez2006contractivity}. However, all $L_p$-distances are jointly convex
\be
d_p\left(\sum_i p_i\rho_i,\sum_i p_i\sigma_i\right)\leqslant\sum_i p_i d_p\left( \rho_i,\sigma_i\right).
\ee
Note that, any power $p>1$ of $d_p$ is also jointly convex.

In addition, we have the \textit{Bures distance} $d_\Bu$ \cite{bures1969extension,spehner2017geometric}
\be\label{Bures}
d_\Bu(\rho,\sigma)\coloneqq\left[2-2\sqrt{F(\rho,\sigma)}\right]^\frac{1}{2},
\ee
where $F(\rho,\sigma)$ is the \textit{Uhlmann fidelity} \cite{uhlmann1976transition} $F(\rho,\sigma)\coloneqq ||\sqrt{\sigma}\sqrt{\rho}||_{1}^2=[\Tr(\sqrt{\rho}\sigma\sqrt{\rho})^\frac{1}{2} ]^2$, which is symmetric. Moreover, if $\rho=\ket{\psi}\bra{\psi}$ and $\sigma=\ket{\phi}\bra{\phi}$, than $F(\rho,\sigma)=|\braket{\psi}{\phi}|^2$. Lastly, we have the \textit{quantum Hellinger distance} $d_\He$ \cite{roga2016geometric,spehner2017geometric}
\be\label{Hellinger}
d_\He(\rho,\sigma)\coloneqq \norm{\sqrt{\sigma}-\sqrt{\rho}}_{2}=\left(2-2\Tr\sqrt{\sigma}\sqrt{\rho} \right)^\frac{1}{2}.
\ee
If $\rho$ and $\sigma$ commute, then $d_\Bu(\rho,\sigma)=d_\He(\rho,\sigma)$. Both the Bures and the Hellinger distances are contractive under CPTP maps. The joint convexity, however, is only satisfied by the square of them, i.e.,
\be
d^2_\Bu\left(\sum_i p_i\rho_i,\sum_i p_i\sigma_i\right)\leqslant\sum_i p_i d^2_\Bu\left( \rho_i,\sigma_i\right),
\ee
with an identical expression for $d_\He$ \cite{spehner2017geometric}. The reader can find a summary of properties about all these distances in Tab. \ref{tab:prop_distances}.

\begin{table}[h]
    \renewcommand{\arraystretch}{1.5}
    \centering
    \begin{tabular}{c c c c c c c c c c }
    \hline
     & $d_\Tr$ & $d_\HS$ & $d^2_\HS$ & $d_p$ & $d_p^p$ & $d_\Bu$ & $d^2_\Bu$ & $d_\He$ & $d^2_\He$ \\
         \hline
    Positive definiteness & \cmark & \cmark & \cmark & \cmark & \cmark & \cmark & \cmark & \cmark & \cmark \\
    Unitary invariance & \cmark & \cmark & \cmark & \cmark & \cmark & \cmark & \cmark & \cmark & \cmark \\
    Joint convexity & \cmark & \cmark & \cmark & \cmark & \cmark & \xmark & \cmark & \xmark & \cmark \\
    Contractivity & \cmark & \xmark & \xmark & \xmark & \xmark & \cmark & \cmark & \cmark & \cmark \\
    \hline
    \end{tabular}
    \caption{Summary of properties that are satisfied by the trace distance $d_\Tr$, the Hilbert-Schmidt distance $d_\HS$, the $L_p$-distances $d_p$ with finite $p> 1$ and $p\neq 2$, the Bures distance $d_\Bu$, the squared Bures distance $d^2_\Bu$, the Hellinger distance $d_\He$, and the squared Hellinger distance $d^2_\He$ between any pair of quantum states.}
    \label{tab:prop_distances}
\end{table}

\section{Results}

\subsection{Geometric conditional information}

Let us propose a geometric measure of conditional information as
\be\label{geo_cond_inf}
I^\square(\mathcal{E|S})_\Omega\coloneqq d_\square^n\left(\Omega,\Omega_\mathcal{S}\otimes\dfrac{\mathbbm{1}_\mathcal{E}}{d_\mathcal{E}}\right),
\ee
where $d_\square$ is any of the distances presented in Sec. \ref{Sec_norms} and $n$ is some power conveniently chosen. This definition is clearly borrowed from the notion of conditional information with divergences we presented in Sec. \ref{sec_inf_divergences} and inspired by the work of Roga et al. \cite{roga2016geometric}. Immediately from the positive definiteness property of distances, we have that
\be
I^\square(\mathcal{E|S})_\Omega=0 \qquad\text{\textit{iff}} \qquad \Omega=\Omega_\mathcal{S}\otimes\frac{\mathbbm{1}_\mathcal{E}}{d_\mathcal{E}}.
\ee
In addition to that, we can expect meaningful measures of conditional information from definition \eqref{geo_cond_inf} once we use distances $d_\square^n$ that respect the contractivity relation \eqref{contractivity}, which is nothing more than a data processing inequality in the geometrical context. This is the case for the trace distance and the square of the Bures and the Hellinger distances. As we are going to see, Hilbert-Schmidt distances can be contractive for unital maps.

\subsection{Summary of main results: geometric monotones of violations of quantum realism}

In order to follow the realism-information relation \eqref{realism_inf_relation}, we must obtain the conditional information of the global states $\Omega_0=\rho_\mc{S}\otimes\ket{e_0}\bra{e_0}$ and $\Omega_t=U_t(\Omega_0)U_t^\dagger$. The first one can be readily calculated through
\be
I^\square(\mc{E|S})_{\Omega_0} = d_\square^n\left(\rho\otimes\ket{e_0}\bra{e_0},\rho\otimes\frac{\mathbbm{1}_\mathcal{E}}{d_\mathcal{E}}\right).
\ee
The conditional information of the global state after unitary interaction $U_\mc{SE}$ will be
\be\label{stinespring}
I^\square(\mc{E|S})_{\Omega_t} =  d_\square^n\left(\Omega_t,\Tr_\mc{E}(\Omega_t)\otimes\frac{\mathbbm{1}_\mathcal{E}}{d_\mathcal{E}}\right).
\ee
Remember that, the realism-information relation \eqref{realism_inf_relation} comes from the fact that measurements establish realism and, therefore, $U_\mc{SE}$ must be such that 
\be
\Tr_\mc{E}(\Omega_t)=\Phi_A(\rho_\mc{S}).
\ee
Because of that, the unitary $U_\mc{SE}$ satisfies [see Theorem 1 of Ref. \cite{orthey2022quantum}]
\be\label{unitary_theorem}
U_\mc{SE}\left(\Phi_A(\rho_\mc{S})\otimes\frac{\mathbbm{1}_\mathcal{E}}{d_\mathcal{E}} \right)U_\mc{SE}^\dagger =\Phi_A(\rho_\mc{S})\otimes\frac{\mathbbm{1}_\mathcal{E}}{d_\mathcal{E}},
\ee
with the condition that $d_\mc{A}=d_\mc{E}$. From the unitary invariance of distances and the above relation, we have that
\be
I^\square(\mc{E|S})_{\Omega_t} =  d_\square^n\left(\rho_\mc{S}\otimes\ket{e_0}\bra{e_0},\Phi_A(\rho_\mc{S})\otimes\frac{\mathbbm{1}_\mathcal{E}}{d_\mathcal{E}}\right).
\ee
Following \eqref{realism_inf_relation}, geometric monotones of VQR (or simply monotones of realism) can be obtained using the following recipe:
\be
\fR_A^\square(\rho_\mc{S})=\fR_A^{\max}(\square)-\Delta I^\square(\mc{E|S})_{\Omega_0\to\Omega_t},
\ee
where $\Delta I^\square(\mc{E|S})_{\Omega_0\to\Omega_t}=I^\square(\mc{E|S})_{\Omega_t}-I^\square(\mc{E|S})_{\Omega_0}$ is the change in the geometric conditional information. Importantly, a \textit{violation of quantum realism} is detected when $\fR_A^\square(\rho_\mc{S})< \fR_A^{\max}(\square)$. After some algebra (see Appendix \ref{ap_calculations}), we reach the following four geometric quantifiers of realism:
\begin{align}
    \fR_A^{\Tr}(\rho_\mc{S})&=\fR_\mathcal{A}^{\max}(\Tr)-\left[d_{\Tr}\left(\rho_\mc{S},\frac{\Phi_A(\rho_\mc{S})}{d_\mathcal{E}}\right)-\frac{d_\mathcal{E}-1}{d_\mathcal{E}}\right];\label{R_Tr}\\
    \fR_A^\HS(\rho_\mc{S})&=\fR_\mathcal{A}^{\max}(\HS)-\frac{1}{d_\mathcal{E}}d^2_\HS(\rho_\mc{S},\Phi_A(\rho_\mc{S}));\label{R_HS}\\
    \fR_A^\Bu(\rho_\mc{S})&=\fR_\mathcal{A}^{\max}(\Bu)-\frac{1}{\sqrt{d_\mathcal{E}}}d^2_\Bu(\rho_\mc{S},\Phi_A(\rho_\mc{S}));\label{R_Bu}\\
    \fR_A^\He(\rho_\mc{S})&=\fR_\mathcal{A}^{\max}(\He)-\frac{1}{\sqrt{d_\mathcal{E}}}d^2_\He(\rho_\mc{S},\Phi_A(\rho_\mc{S})).\label{R_He}
\end{align}
Remember that $d_\mc{E}=d_\mc{A}$. We have chosen $n=1$ for the realism quantifier based on the trace distance and $n=2$ for the other cases. The max value $\fR_\mathcal{A}^{\max}(\square)$ can be obtained by optimization algorithms, but the analytical solution can be found by making $\rho_\mc{S}$ a maximally entangled state of two qudits in each case. In Fig. \ref{fig:1}(a) we plot $\fR_A^{\max}$ for Trace, Hilbert-Schmidt, Bures, and Hellinger distances as a function of the dimension $d_\mc{E}$ for a $d_\mc{E}$-outcome observable $A=\{A_a\}_{a=0}^{d_\mc{E}-1}$ in the computational basis.

An example can be suitable:

\begin{example}\label{ex_RGwerner}
Let us see how the realism of a spin-1/2 observable $A=\hat{v}\cdot\Vec{\sigma}$ regarding the first particle of the Werner state 
\be\label{werner}
\rho_\epsilon=(1-\epsilon)\frac{\mathbbm{1}_4}{4}+\epsilon\ket{\phi^+}\bra{\phi^+},
\ee
where $\ket{\phi^+}=(\ket{00}+\ket{11})\sqrt{2}$, increases with the loss of coherence within the quantifiers mentioned above. The calculations were made with the help of symbolic computation software and the lengthy expressions will be omitted, but the results can be visualized in Fig. \ref{fig:1} (b). As expected for the Werner state, the realism degree does not depend on the direction of measurement $\hat{v}$, because the maximally entangled state of two qubits is rotationally invariant. Notably, $\fR_A^{\Tr}$ violates Axiom \ref{A_reality_meas} for $\epsilon\in[0,1/3]$, which makes the trace distance unsuitable for monotones of realism. It is also important to note that since $\rho_\epsilon$ commutes with $\Phi_A(\rho_\epsilon)$, then we have $\fR_A^\Bu(\rho_\epsilon)=\fR_A^\He(\rho_\epsilon)$.
\end{example}

\begin{figure}[t]
    \centering
    \includegraphics[width=0.9\linewidth]{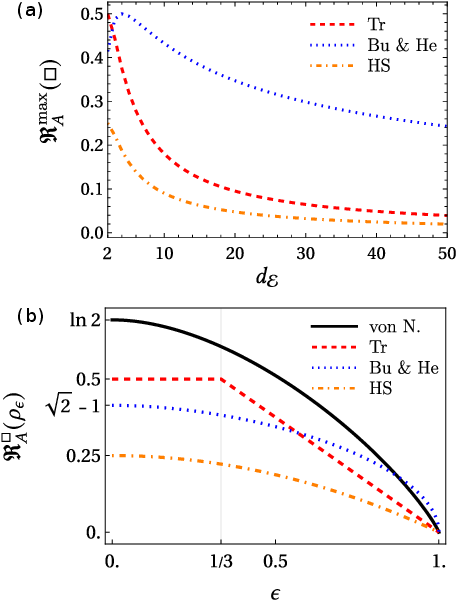}
    \caption{(a) $\fR_A^{\max}$ for Trace (red dashed line), Hilbert-Schmidt (orange dash-dotted line), and Bures and Hellinger (blue dotted line) distances as a function of the dimension $d_\mc{E}$. (b) Geometric quantifiers of VQR $\fR_A^\square$ of any spin-1/2 observable regarding the first particle of a Werner state $\rho_\epsilon$ from Example \ref{ex_RGwerner} as a function of $\epsilon$ for von Neumann \eqref{realism_vonNeumann} (black line), Tr \eqref{R_Tr} (red dashed line), Bu \eqref{R_Bu} and He \eqref{R_He} (blue dotted line), and HS (orange dash-dotted line) \eqref{R_HS}. $\fR_A^{\max}(\square)$ for $d_\mc{E}=2$ is highlighted in the vertical axis.}
    \label{fig:1}
\end{figure}

We can also propose to measure the realism degree of an observable by means of the $L_p$-distances \eqref{Lp_distance} for $p>1$ and $p\neq 2$. Such a quantifier could generalize the trace and the Hilbert-Schmidt realities. However, the expression we have obtained for it does not provide any further insight [see Eq. \eqref{delta_Lp} in Appendix \ref{ap_calculations}]. Nonetheless, we present in Tab. \ref{tab:axioms_distances} which axioms we know that are fulfilled by $\mf{R}_A^p$.

Another interesting aspect of quantifiers \eqref{R_Tr}-\eqref{R_He} induced by the geometric distances is that $\fR_A^{\max}(\square)$ falls asymptotically to zero with higher $d_\mc{E}$, eventually signalling full realism for any state and any observable. This behaviour does not happen with entropic quantifiers such as \eqref{realism_vonNeumann}, in which case we have $\fR_A^{\max}(\text{von N.})=\ln d_\mc{E}$.

Let us now investigate which of the four axioms mentioned in Sec. \ref{sec_realism_conditional} are satisfied by each one of the four geometric quantifiers of realism in Eqs. \eqref{R_Tr}-\eqref{R_He}.

\subsection{VQR with trace distance}

Example \ref{ex_RGwerner} alone is sufficient to disqualify $\fR_A^{\Tr}$ as a bona fide monotone for VQR, as it violates Axiom \ref{A_reality_meas}. Specifically, maximum realism is achieved by states other than $\Phi_A(\rho)$. Indeed, the VQR quantifier based on the trace distance exhibits a sudden death of irrealism for any spin-1/2 observable associated with either particle of the Werner state when $\epsilon \leqslant 1/3$ (see Fig. \ref{fig:1}). Although the trace distance satisfies the property of positive definiteness, its application in Eq. \eqref{R_Tr} involves one entry that is non-normalized, namely, $\Phi_A(\rho_\mc{S})/d_\mc{E}$. Consequently, deriving a VQR monotone via the realism-information relation \eqref{realism_inf_relation} leads to a loss of Axioms \ref{A_reality_meas} and \ref{A_uncertainty}. Nevertheless, it remains plausible to establish a quantum resource theory of irrealism induced by the trace distance. In such a theory, even though states with $0 < \epsilon \leqslant 1/3$ violate $\Phi_A(\rho_\epsilon) = \rho_\epsilon$, they might lack sufficient resource to qualify as resourceful states. Nonetheless, the trace distance is invariant under the addition of an uncorrelated system, contractive under the discard of parts of the system, and also jointly convex, therefore satisfying Axioms \ref{A_role_parts} and \ref{A_mix}.

\subsection{VQR with Hilbert-Schmidt distance}

Although the Hilbert-Schmidt distance is not contractive under CPTP maps in general, it becomes contractive under CPTP \textit{unital} maps, that is, maps that keep the identity invariant. That is the case for projective measurements $\Phi_A$, monitorings $\mathcal{M}_A^\epsilon$, and partial traces. If $\mathcal{C}$ is a CPTP unital map, then $\norm{\mathcal{C}(\rho)}_p\leqslant\norm{\rho}_p$, where equality holds if $\rho$ is a maximally mixed state. Also, equality holds if $p=1$. In the particular case that $\mathcal{C}$ is a non-selective measurement (e.g., $\mathcal{C}=\Phi_A$), then equality hols \textit{iff} $\rho=\mathcal{C}(\rho)$ (see Theorem 5.2 in Ref. \cite{krein1969}). In addition to that, the square of $d_\HS$ remains contractive because $x\mapsto x^2$ is an injective monotonically increasing function for $x\geqslant 0$. Therefore, quantifier $\fR_A^\text{HS}$ in \eqref{R_HS} satisfies Axiom \ref{A_role_parts}(a). It is easy to see that Axioms \ref{A_reality_meas}, \ref{A_uncertainty}, and \ref{A_mix} are also valid given the positive definiteness and the joint convexity properties of the $L_p$-distances. The irrelevance of the uncorrelated [Axiom \ref{A_role_parts}(b)], however, is not fulfilled because
\begin{align}
d^2_\HS(\rho\otimes\sigma,\Phi_A(\rho)\otimes\sigma)&=\norm{\rho\otimes\sigma-\Phi_A(\rho)\otimes\sigma}_{2}^2,\nonumber\\
&=\norm{\rho-\Phi_A(\rho)}_{2}^2\norm{\sigma}_{2}^2,\nonumber\\
&> d^2_\HS(\rho,\Phi_A(\rho)),\label{HS_uncorrelated}
\end{align}
if $\sigma$ is not pure. That means that a quantifier of VQR induced by the Hilbert-Schmidt distance fails to ignore completely uncorrelated systems when assessing the degree of realism in a given context. In fact, this is a common issue regarding quantum correlations derived from the Hilbert-Schmidt distance, as argued by Piani \cite{piani2012problem}. Having said that, the quantifier $\fR_A^\HS$ does not meet the minimum requirements to be considered a monotone of VQR.

\subsection{VQR with Bures and Hellinger distances}

When we try to derive quantifiers for VQR based on the realism-information relation \eqref{realism_inf_relation}, the Bures and Hellinger distances are the only ones---of the ones we have seen so far---that provide us with bona-fide monotones, namely $\fR_A^\Bu$ and $\fR_A^\He$ from Eqs. \eqref{R_Bu} and \eqref{R_He}. This is due to the positive definiteness, contractivity, and joint convexity properties fulfilled by these distances. In addition to that, these distances have a close connection with entropic quantities known as R\'enyi divergence \cite{petz1986quasi}
\be\label{renyi}
D_\alpha(\rho||\sigma)\coloneqq\frac{1}{\alpha-1}\ln\frac{\Tr\left(\rho^\alpha\sigma^{1-\alpha}\right)}{\Tr\rho},
\ee
for $\alpha\in(0,1)\cup(1,+\infty)$ and sandwiched R\'enyi divergence \cite{muller2013quantum,wilde2014strong}
\be\label{sandwiched}
\tD_\alpha(\rho||\sigma)\coloneqq\frac{1}{\alpha-1}\ln\left\{\frac{1}{\Tr\rho}\Tr\left[\left(\sigma^\frac{1-\alpha}{2\alpha}\rho\sigma^\frac{1-\alpha}{2\alpha} \right)^\alpha \right]\right\},
\ee
with the same conditions of quantity \eqref{renyi}. Equations \eqref{renyi} and \eqref{sandwiched} are said to be generalizations of Eq. \eqref{relative_entropy} because $D_{\alpha\to 1}(\rho||\sigma)=\tD_{\alpha\to 1}(\rho||\sigma)=S(\rho||\sigma)$. One can see that
\begin{align}
d^2_\Bu(\rho,\sigma)&=2-2\exp\left(-\frac{1}{2}\tD_{\alpha=\frac{1}{2}}(\rho||\sigma) \right),\label{d_Bu_tD}\\
d^2_\He(\rho,\sigma)&=2-2\exp\left(-\frac{1}{2}D_{\alpha=\frac{1}{2}}(\rho||\sigma) \right),\label{d_He_D}
\end{align}
which means that the Bures and the Hellinger realism monotones can be written as injective increasing functions of special cases of the R\'enyi divergences \eqref{renyi} and \eqref{sandwiched}, precisely when $\alpha=1/2$. Since monotones of VQR induced by R\'enyi and sandwiched R\'enyi divergences were successfully obtained in \cite{orthey2022quantum}, we have that \eqref{d_Bu_tD} and \eqref{d_He_D} also indicates that $\fR_A^\Bu$ and $\fR_A^\He$ satisfy the requirements to be framed as bona-fide monotones of VQR.

The main difference between the Bures and the Hellinger distances is that the latter is sensitive to the non-commutativity of states. In fact, the role of this feature in the assessment of the VQR of observables is not clear to us---which suggests a path of research in line with the work of Martins et al. \cite{martins2020quantum}. Nevertheless, the reader can get an idea of how close these quantifiers are to each other by the following example.

\begin{figure*}[t]
    \centering
    \includegraphics[width=0.8\linewidth]{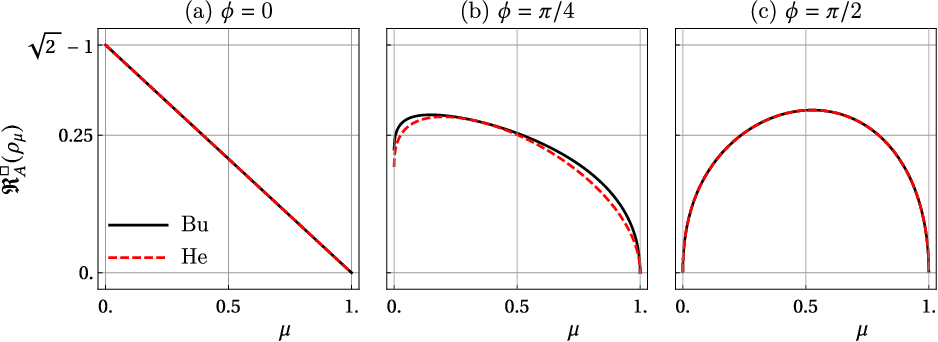}
    \caption{Monotones of VQR $\mf{R}_A^\Bu$ (black line) and $\mf{R}_A^\He$ (red dashed line) of a spin-$1/2$ observable $A=\hat{v}\cdot\vec{\sigma}$, where $\hat{v}=(\cos\theta\sin\phi,\sin\theta\sin\phi,\cos\phi)$, regarding either qubit of the state $\rho_\mu$ \eqref{mu_state} for (a) $\phi=0$, (b) $\phi=\pi/4$, and (c) $\phi=\pi/2$.}
    \label{fig:RGalpha}
\end{figure*}

\begin{example}
Consider the two-qubit state 
\be\label{mu_state}
\rho_\mu=\frac{\mbb{1}}{4}+\frac{\mu}{4}\left(\sigma_x\otimes\sigma_x-\sigma_y\otimes\sigma_y \right) + \frac{2\mu-1}{4}\sigma_z\otimes\sigma_z,
\ee
where $\mu\in[0,1]$. Particularly, $\rho_{\mu=1}=\ket{\phi^ +}\bra{\phi^+}$, where $\ket{\phi^+}=(\ket{00}+\ket{11})/\sqrt{2}$, and  $\rho_{\mu=0}=(\ket{01}\bra{01}+\ket{10}\bra{10})/2$. Now, consider the observable
\be
A=\hat{u}\cdot\vec{\sigma}=\cos\theta\sin\phi\sigma_x+\sin\theta\sin\phi\sigma_y+\cos\phi\sigma_z,
\ee
where $\theta$ and $\phi$ and the polar and azimuthal angles, and $\vec{\sigma}$ is the Pauli vector. When applied to either qubit, map $\Phi_A$ results in a state $\Phi_A(\rho_\mu)$ that does not always commute with $\rho_\mu$. After some algebra, one can see that $\fR_A^\Bu(\rho_\mu)$ and $\fR_A^\He(\rho_\mu)$ are dependent only on $\mu$ and $\phi$. The slight differences between these two monotones of VQR due to the non-commutativity between $\rho_\mu$ and $\Phi_A(\rho_\mu)$ can be seen in Fig. \ref{fig:RGalpha}.
\end{example}

In Table \ref{tab:axioms_distances}, we summarize our results with a list of axioms that a quantifier needs to satisfy to be framed as a bona-fide monotone of VQR.

\begin{table}[h]
    \renewcommand{\arraystretch}{1.5}
    \centering
    \begin{tabular}{l c c c c c}
    \hline
    & $\mf{R}_A^{\Tr}$  & $\mf{R}_A^\HS$ & $\mf{R}_A^p$ & $\mf{R}_A^\Bu$ &  $\mf{R}_A^\He$ \\    
         \hline
    Axiom \ref{A_reality_meas} (measurements) & \xmark & \cmark & ?  & \cmark & \cmark \\
    Axiom \ref{A_role_parts}(a) (part discard) & \cmark & \cmark & ? & \cmark & \cmark \\
    Axiom \ref{A_role_parts}(b) (uncorrelated part) & \cmark & \xmark & \xmark & \cmark & \cmark\\
    Axiom \ref{A_uncertainty} (uncertainty relation) & \xmark & \cmark & ? & \cmark & \cmark\\
    Axiom \ref{A_mix} (mixing) & \cmark & \cmark & \cmark & \cmark & \cmark\\
    \hline
    \end{tabular}
    \caption{Summary of the axioms satisfied by the $A$-realism quantifiers $\fR_A^{\Tr}$ [Eq.~\eqref{R_Tr}], $\fR_A^\HS$ [Eq.~\eqref{R_HS}], $\fR_A^{p}$, $\fR_A^{\Bu}$ [Eq. \eqref{R_Bu}], and $\fR_A^{\He}$ [Eq. \eqref{R_He}] built out of their corresponding distance measures $d_{\Tr}$, $d^2_\HS$, $d_p^p$ (for finite $p>1$ and $p\neq 2$), $d^2_\Bu$, and $d^2_\He$, whose properties are listed in Table \ref{tab:prop_distances}. Our approach legitimates just two geometric realism monotones, namely, $\fR_A^{\Bu}$ and $\fR_A^{\He}$.}
    \label{tab:axioms_distances}
\end{table}

\section{Discussion}

In this work, we have extended the framework for quantifying violations of quantum realism (also known as quantum irrealism \cite{bilobran2015measure}) by exploring geometric distances as potential tools for constructing monotones of VQR. Geometric distances, due to their symmetry and computational simplicity, provide an intriguing alternative to relative entropies. Among the distances analyzed, we identified Bures and Hellinger distances as satisfying all minimal criteria for bona fide monotones of VQR, while trace distance, Hilbert-Schmidt distance, and $L_p$-distances in general fail to attend those criteria.

The realism-information relation \eqref{realism_inf_relation} proposed in \cite{orthey2022quantum} imposes that quantifiers satisfying the Axioms \eqref{A_reality_meas}-\eqref{A_mix} for VQR should inherently align with informational structures. Specifically, this relation sets a connection between conditional information \eqref{cond_information_relative} and the emergence of realism in quantum systems \eqref{stinespring}, therefore guiding the construction of monotones. The failure of certain distances, such as trace distance and Hilbert-Schmidt distance, to meet the criteria for monotones stresses the requirements for valid quantifiers of VQR. Specifically, the inability of these distances to consistently capture VQR when $\Phi_A(\rho)\neq\rho$ [see the plateau in Fig. \ref{fig:1}(b)] or their counterintuitive behaviour under the addition of completely uncorrelated states [see \eqref{HS_uncorrelated}] highlights the importance of aligning distance-based monotones with physical principles of quantum realism.

Our findings suggest that quantifiers of VQR that are not connected with entropic measures are not suitable for measuring violations of quantum realism since only the Bures and Hellinger distances could provide bona fide monotones of VQR. Indeed, these distances can be expressed as functions of R\'enyi \eqref{d_He_D} and Sandwiched R\'enyi divergences \eqref{d_Bu_tD} in the specific cases when they are symmetric, which reinforces the connection of Bures and Hellinger distances to informational measures.

Clearly, the set of distances we have explored in this work is not exhaustive. An intriguing open question is whether other geometric distances exist that satisfy all the axioms for valid VQR monotones. Identifying such distances could broaden the toolbox for quantifying violations of quantum realism and potentially reveal deeper connections between geometry and quantum information theory. Additionally, it remains to be seen whether these potential candidates share a fundamental link with entropic measures, similar to the Bures and Hellinger distances. Investigating this relationship could further elucidate the interplay between geometric and informational perspectives on quantum realism, offering new insights into the structure and interpretation of quantum measurements.

An important point to highlight is that $\fR_A^{\max}(\square)$ for the Tr, HS, Bu, and He cases approaches zero as $d_\mc{E}$ increases, as shown in Fig. \ref{fig:1}(a). This behavior contrasts with the original entropic measure of VQR given by \eqref{realism_vonNeumann}, whose maximum value grows logarithmically with $d_\mc{E}$. This observation raises several open questions: Is it essential for a measure of realism to assign higher values to observables in higher-dimensional systems? Should the realism-information relation \eqref{realism_inf_relation} be rescaled to ensure that the geometric monotones exhibit increasing maximum values of realism as the dimension grows? These questions highlight potential directions for refining the theoretical framework of quantum realism.

In conclusion, this study advances the understanding of VQR by identifying geometric distances that align with the informational principles underlying quantum realism. Our results provide both theoretical insights and practical tools for analyzing the interplay between measurement, information, and notions of realism in quantum mechanics, offering a promising avenue for future research in quantum foundations and information science.

\begin{acknowledgements}
ACO thanks R. M. Angelo and Danilo Fucci for insightful conversations. The authors acknowledge the support from the Polish National Science Center (grant No. 2022/46/E/ST2/00115).
\end{acknowledgements}

\bibliography{bibliography.bib}

\appendix

\section{Auxiliary lemmas}

The proof of the following Lemma can be found in Appendix A of \cite{costa2013bayes}.
\begin{lemma}\label{lemma_f}
For any function $f$, any quantum states $\rho$ and $\sigma$, and any observable $A$, we have
\be
\Tr[\rho\, f(\Phi_A(\sigma))]=\Tr[\Phi_A(\rho)f(\Phi_A(\sigma))],
\ee
where $\Phi_A$ is the map of non-selective projective measurements given by $\Phi_A(\rho)=\sum_a (A_a\otimes\mbb{1}_\mc{B})\rho(A_a\otimes\mbb{1}_\mc{B})$, where $A_a=\ket{a}\bra{a}$ are projectors that sum up to identity in $\mc{H_A}$.
\end{lemma}

\begin{lemma}\label{lemma_HS}
The Hilbert-Schmidt norm $\norm{\cdot}_2=\sqrt{\Tr|\cdot|^2}$ satisfies 
\be
\norm{\rho}_2^2-\norm{\Phi_A(\rho)}_2^2=\norm{\rho-\Phi_A(\rho)}_2^2
\ee
for any state $\rho$ and any non-selective projective measurement map $\Phi_A$.
\end{lemma}
\begin{proof}
The proof relies in the direct application of Lemma \ref{lemma_f}:
\begin{align}
    &\norm{\rho-\Phi_A(\rho)}_2^2\nonumber\\
    &= \Tr\left(\rho-\Phi_A(\rho) \right)^2,\\
    &=\Tr\rho^2-2\Tr\rho\Phi_A(\rho)+\Tr\Phi_A(\rho)^2,\\
    &=\Tr\rho^2-2\Tr\rho\Phi_A(\rho)+\Tr\rho\Phi_A(\rho),\qquad\text{(Lemma \ref{lemma_f})}\\
    &=\Tr\rho^2-\Tr\rho\Phi_A(\rho),\\
    &=\Tr\rho^2-\Tr\Phi_A(\rho)^2,\qquad\text{(Lemma \ref{lemma_f})}\\
    &=\norm{\rho}_2^2-\norm{\Phi_A(\rho)}_2^2.
\end{align}
\end{proof}

\section{Calculations}\label{ap_calculations}

For $\Omega_0=\rho\otimes\mathbbm{1}_\mathcal{E}/d_\mathcal{E}$, the Eq. \eqref{geo_cond_inf} reads:
\begin{align}
        I^\square_\mathcal{E|S}(\Omega_0)&=d_\square^n\left(\Omega_0,\Tr_\mathcal{E}\Omega_0\otimes\frac{\mathbbm{1}_\mathcal{E}}{d_\mathcal{E}}\right),\\
        &=d_\square^n\left(\rho\otimes\ket{e_0}\bra{e_0},\rho\otimes\frac{\mathbbm{1}_\mathcal{E}}{d_\mathcal{E}}\right).\label{IES0}
\end{align}
After the unitary evolution, we have:
\begin{align}
        I^\square_\mathcal{E|S}(\Omega_t)&=d_\square^n\left(\Omega_t,\Tr_\mathcal{E}\Omega_t\otimes\frac{\mathbbm{1}_\mathcal{E}}{d_\mathcal{E}}\right),\\
        &=d_\square^n\left(U_t\left(\rho\otimes\ket{e_0}\bra{e_0}\right)U_t^\dagger,\Phi_A(\rho)\otimes\frac{\mathbbm{1}_\mathcal{E}}{d_\mathcal{E}}\right),\\
        &=d_\square^n\left(U_t\left(\rho\otimes\ket{e_0}\bra{e_0}\right)U_t^\dagger,U_t\left(\Phi_A(\rho)\otimes\frac{\mathbbm{1}_\mathcal{E}}{d_\mathcal{E}}\right)U_t^\dagger\right),\label{eqB5}\\
        &=d_\square^n\left(\rho\otimes\ket{e_0}\bra{e_0},\Phi_A(\rho)\otimes\frac{\mathbbm{1}_\mathcal{E}}{d_\mathcal{E}}\right).\label{IESt}
\end{align}
Note that we have used \eqref{unitary_theorem} in \eqref{eqB5}. Now, let us evaluate the above expressions for the following specific cases:

\subsection{$L_p$-distances}

The $L_p$-distance is defined as \cite{bhatia2013matrix,spehner2017geometric}
\be
d_p(\rho,\sigma)\coloneqq\norm{\sigma-\rho}_p=\left(\Tr|\sigma-\rho|^p \right)^\frac{1}{p}
\ee
for all $p\geqslant 1$. Therefore, \eqref{IES0} becomes
\begin{align}
    I^p_\mathcal{E|S}(\Omega_0)&=d_p^p\left(\rho\otimes\ket{e_0}\bra{e_0},\rho\otimes\frac{\mathbbm{1}_\mathcal{E}}{d_\mathcal{E}}\right),\\
    &=\bignorm{\rho\otimes\ket{e_0}\bra{e_0}-\rho\otimes\frac{\mathbbm{1}_\mathcal{E}}{d_\mathcal{E}}}_p^p,\\
    &=\norm{\rho}_p^p\bignorm{\ket{e_0}\bra{e_0}-\frac{\mathbbm{1}_\mathcal{E}}{d_\mathcal{E}}}_p^p,\\
    &=\norm{\rho}_p^p\bignorm{\sum_{k=0}^{d_\mathcal{E}-1} \left(\delta_{k,0}-\frac{1}{d_\mathcal{E}}\right)\ket{e_k}\bra{e_k} }_p^p,\\
    &=\norm{\rho}_p^p\left\{\Tr\left|\sum_{k=0}^{d_\mathcal{E}-1} \left(\delta_{k,0}-\frac{1}{d_\mathcal{E}}\right)\ket{e_k}\bra{e_k}\right|^p \right\},\\
    &=\norm{\rho}_p^p\left\{\sum_{k=0}^{d_\mathcal{E}-1} \left|\delta_{k,0}-\frac{1}{d_\mathcal{E}}\right|^p\Tr\ket{e_k}\bra{e_k} \right\},\\
    &=\norm{\rho}_p^p\left\{\sum_{k=0}^{d_\mathcal{E}-1} \left|\delta_{k,0}-\frac{1}{d_\mathcal{E}}\right|^p \right\},\\
    &=\norm{\rho}_p^p\left\{\left(1-\frac{1}{d_\mathcal{E}}\right)^p+\frac{d_\mathcal{E}-1}{d_\mathcal{E}^p} \right\},\\
    &=\norm{\rho}_p^p\frac{1}{d_\mathcal{E}^p}\left[\left(d_\mathcal{E}-1 \right)^p+d_\mathcal{E}-1 \right].
\end{align}
Similarly, for \eqref{IESt} we have 
\begin{align}
     I^p_\mathcal{E|S}(\Omega_t)&=d_p^p\left(\rho\otimes\ket{e_0}\bra{e_0},\Phi_A(\rho)\otimes\frac{\mathbbm{1}_\mathcal{E}}{d_\mathcal{E}}\right),\\
    &=\bignorm{\rho\otimes\ket{e_0}\bra{e_0}-\Phi_A(\rho)\otimes\frac{\mathbbm{1}_\mathcal{E}}{d_\mathcal{E}}}_p^p,\\
    &=\bignorm{\sum_{k=0}^{d_\mathcal{E}-1}\left(\delta_{k,0}\rho-\dfrac{\Phi_A(\rho)}{d_\mathcal{E}}\right)\otimes\ket{e_k}\bra{e_k}}_p^p,\\
    &=\Tr\left|\sum_{k=0}^{d_\mathcal{E}-1}\left(\delta_{k,0}\rho-\frac{\Phi_A(\rho)}{d_\mathcal{E}}\right)\otimes\ket{e_k}\bra{e_k} \right|^p ,\\
    &=\sum_{k=0}^{d_\mathcal{E}-1}\Tr\left|\delta_{k,0}\rho-\frac{\Phi_A(\rho)}{d_\mathcal{E}}\right|^p\Tr\ket{e_k}\bra{e_k}  ,\\
    &=\sum_{k=0}^{d_\mathcal{E}-1}\Tr\left|\delta_{k,0}\rho-\frac{\Phi_A(\rho)}{d_\mathcal{E}}\right|^p  ,\\
    &=\Tr\left|\rho-\frac{\Phi_A(\rho)}{d_\mathcal{E}}\right|^p+\left(d_\mathcal{E}-1 \right)\Tr\left|\frac{\Phi_A(\rho)}{d_\mathcal{E}}\right|^p  ,\\
    &=d_p^p\left(\rho,\frac{\Phi_A(\rho)}{d_\mathcal{E}}\right)+\frac{d_\mathcal{E}-1}{d_\mathcal{E}^p} \norm{\Phi_A(\rho)}^p_p
\end{align}
Finally, the change in conditional information induced by the $L_p$-distances is given by:
\begin{align}
&\Delta I_\mc{E|S}^p(\Omega_t,\Omega_0)\nonumber\\
&=d_p^p\left( \rho,\frac{\Phi_A(\rho)}{d_\mc{E}} \right)-\frac{d_\mc{E}-1}{d_\mc{E}^p}\left(\norm{\Phi_A(\rho)}_p^p-\norm{\rho}_p^p \right)\\
&-\norm{\rho}_p^p\left(\frac{d_\mc{E}-1}{d_\mc{E}} \right)^p \label{delta_Lp}   
\end{align}

In what follows, we derive expressions for $\Delta I_\mathcal{E|S}^{\Tr}(\Omega_t,\Omega_0)$ for the Trace distance ($p=1$)
\begin{align}
    I_\mathcal{E|S}^{\Tr}(\Omega_0)&=2\left(\frac{d_\mathcal{E}-1}{d_\mathcal{E}}\right).\\
    I_\mathcal{E|S}^{\Tr}(\Omega_t)&=d_{\Tr}\left(\rho,\frac{\Phi_A(\rho)}{d_\mathcal{E}}\right)+\frac{d_\mathcal{E}-1}{d_\mathcal{E}}.\\
    \Delta I_\mathcal{E|S}^{\Tr}(\Omega_t,\Omega_0)&=d_{\Tr}\left(\rho,\frac{\Phi_A(\rho)}{d_\mathcal{E}}\right)-\frac{d_\mathcal{E}-1}{d_\mathcal{E}},
\end{align}
and the Hilbert-Schmidt distance ($p=2$)
\begin{align}
    I_\mathcal{E|S}^\HS(\Omega_0)&=\frac{d_\mathcal{E}-1}{d_\mathcal{E}}\norm{\rho}_2^2.\\
    I_\mathcal{E|S}^\HS(\Omega_t)&=\Tr\left(\rho-\frac{\Phi_A(\rho)}{d_\mathcal{E}}\right)^2+\frac{d_\mathcal{E}-1}{d_\mathcal{E}^2}\Tr\Phi_A(\rho)^2 ,\\
    &=\Tr\rho^2-\frac{2}{d_\mathcal{E}}\Tr\rho\Phi_A(\rho)\nonumber\\
    &+\frac{1}{d_\mathcal{E}^2}\Tr\Phi_A(\rho)^2+\frac{d_\mathcal{E}-1}{d_\mathcal{E}^2}\Tr\Phi_A(\rho)^2 ,\\
    &= \Tr\rho^2-\frac{2}{d_\mathcal{E}}\Tr\Phi_A(\rho)^2+\frac{1}{d_\mathcal{E}}\Tr\Phi_A(\rho)^2,\label{use1_lemmaHS}\\
    &=\norm{\rho}_2^2-\frac{1}{d_\mathcal{E}}\norm{\Phi_A(\rho)}_2^2.\\
    \Delta I_\mathcal{E|S}^\HS(\Omega_t,\Omega_0)&= \frac{1}{d_\mathcal{E}}\left(\norm{\rho}_2^2-\norm{\Phi_A(\rho)}_2^2\right),\\
    &=\frac{1}{d_\mathcal{E}}\norm{\rho-\Phi_A(\rho)}_2^2,\label{use2_lemmaHS}\\
    &=\frac{1}{d_\mathcal{E}}d^2_\HS(\rho,\Phi_A(\rho)).
\end{align}
Note that we used Lemma \ref{lemma_HS} in Eqs. \eqref{use1_lemmaHS} and \eqref{use2_lemmaHS}.

\subsection{Bures distance}
For the Bures distance, first we need the Uhlmann fidelity \cite{uhlmann1976transition} given by $F(\rho,\sigma)\coloneqq ||\sqrt{\sigma}\sqrt{\rho}||_{1}^2=[\Tr(\sqrt{\rho}\sigma\sqrt{\rho})^\frac{1}{2} ]^2$. The fidelity between $\rho\otimes\ket{e_0}\bra{e_0}$ and $\rho\otimes\frac{\mathbbm{1}_\mathcal{E}}{d_\mathcal{E}}$ is
\begin{align}
    &F\left( \rho\otimes\ket{e_0}\bra{e_0},\rho\otimes\frac{\mathbbm{1}_\mathcal{E}}{d_\mathcal{E}}\right)\nonumber\\
    &= \bignorm{\sqrt{\rho\otimes\ket{e_0}\bra{e_0}}\sqrt{\rho\otimes\frac{\mathbbm{1}_\mathcal{E}}{d_\mathcal{E}}}}_1^2,\\
    &=\bignorm{\left(\sqrt{\rho}\otimes\ket{e_0}\bra{e_0}\right)\left(\sqrt{\rho}\otimes\frac{\mathbbm{1}_\mathcal{E}}{\sqrt{d_\mathcal{E}}} \right)}_1^2,\\
    &=\bignorm{\rho\otimes\frac{\ket{e_0}\bra{e_0}}{\sqrt{d_\mathcal{E}}}}_1^2,\\
    &=\frac{1}{d_\mathcal{E}}\norm{\rho}_1^2\norm{\ket{e_0}\bra{e_0}}_1^2,\\
    &=\frac{1}{d_\mathcal{E}}.
\end{align}
Similarly, the fidelity between $\rho\otimes\ket{e_0}\bra{e_0}$ and $\Phi_A(\rho)\otimes\frac{\mathbbm{1}_\mathcal{E}}{d_\mathcal{E}}$ can be calculated by
\begin{align}
    &F\left( \rho\otimes\ket{e_0}\bra{e_0},\Phi_A(\rho)\otimes\frac{\mathbbm{1}_\mathcal{E}}{d_\mathcal{E}}\right)\nonumber\\
    &= \bignorm{\sqrt{\rho\otimes\ket{e_0}\bra{e_0}}\sqrt{\Phi_A(\rho)\otimes\frac{\mathbbm{1}_\mathcal{E}}{d_\mathcal{E}}}}_1^2,\\
    &=\bignorm{\left(\sqrt{\rho}\otimes\ket{e_0}\bra{e_0}\right)\left(\sqrt{\Phi_A(\rho)}\otimes\frac{\mathbbm{1}_\mathcal{E}}{\sqrt{d_\mathcal{E}}} \right)}_1^2,\\
    &=\bignorm{\sqrt{\rho}\sqrt{\Phi_A(\rho)}\otimes\frac{\ket{e_0}\bra{e_0}}{\sqrt{d_\mathcal{E}}}}_1^2,\\
    &=\frac{1}{d_\mathcal{E}}\norm{\sqrt{\rho}\sqrt{\Phi_A(\rho)}}_1^2\norm{\ket{e_0}\bra{e_0}}_1^2,\\
    &=\frac{1}{d_\mathcal{E}}F(\rho,\Phi_A(\rho)).
\end{align}
Now, we can calculate:
\begin{align}
    I_\mathcal{E|S}^\Bu(\Omega_0)&=d^2_\Bu\left(\rho\otimes\ket{e_0}\bra{e_0},\rho\otimes\frac{\mathbbm{1}_\mathcal{E}}{d_\mathcal{E}}\right),\\
    &=2-2\sqrt{F\left( \rho\otimes\ket{e_0}\bra{e_0},\rho\otimes\frac{\mathbbm{1}_\mathcal{E}}{d_\mathcal{E}}\right)},\\
    &=2-\frac{2}{\sqrt{d_\mathcal{E}}}.
\end{align}
\begin{align}
    I_\mathcal{E|S}^\Bu(\Omega_t)&=d^2_\Bu\left(\rho\otimes\ket{e_0}\bra{e_0},\Phi_A(\rho)\otimes\frac{\mathbbm{1}_\mathcal{E}}{d_\mathcal{E}}\right),\\
    &=2-2\sqrt{F\left( \rho\otimes\ket{e_0}\bra{e_0},\Phi_A(\rho)\otimes\frac{\mathbbm{1}_\mathcal{E}}{d_\mathcal{E}}\right)},\\
    &=2-2\sqrt{\frac{1}{d_\mathcal{E}} F\left(\rho,\Phi_A(\rho)\right)}.
\end{align}
\begin{align}
    \Delta I_\mathcal{E|S}^\Bu(\Omega_t,\Omega_0)&=\frac{1}{\sqrt{d_\mathcal{E}}}\left(2-2\sqrt{F(\rho,\Phi_A(\rho))} \right),\\
    &=\frac{1}{\sqrt{d_\mathcal{E}}}d^2_\Bu(\rho,\Phi_A(\rho)).
\end{align}

\subsection{Hellinger distance}

The quantum Hellinger distance $d_\He$ \cite{roga2016geometric,spehner2017geometric} is given by $
d_\He(\rho,\sigma)\coloneqq \norm{\sqrt{\sigma}-\sqrt{\rho}}_{2}=\left(2-2\Tr\sqrt{\sigma}\sqrt{\rho} \right)^\frac{1}{2}$. Therefore,
\begin{align}
    I_\mathcal{E|S}^\He(\Omega_0)&=d^2_\He\left(\rho\otimes\ket{e_0}\bra{e_0},\rho\otimes\frac{\mathbbm{1}_\mathcal{E}}{d_\mathcal{E}}\right),\\
    &=2-2\Tr\left(\sqrt{\rho\otimes\ket{e_0}\bra{e_0}}\sqrt{\rho\otimes\frac{\mathbbm{1}_\mathcal{E}}{d_\mathcal{E}}} \right),\\
    &=2-2\Tr\left(\rho\otimes\frac{\ket{e_0}\bra{e_0}}{\sqrt{d_\mathcal{E}}}\right),\\
    &=2-\frac{2}{\sqrt{d_\mathcal{E}}}\Tr(\rho)\Tr(\ket{e_0}\bra{e_0}),\\
    &=2-\frac{2}{\sqrt{d_\mathcal{E}}}.
\end{align}
\begin{align}
    I_\mathcal{E|S}^\He(\Omega_t)&=d^2_\He\left(\rho\otimes\ket{e_0}\bra{e_0},\Phi_A(\rho)\otimes\frac{\mathbbm{1}_\mathcal{E}}{d_\mathcal{E}}\right),\\
    &=2-2\Tr\left(\sqrt{\rho\otimes\ket{e_0}\bra{e_0}}\sqrt{\Phi_A(\rho)\otimes\frac{\mathbbm{1}_\mathcal{E}}{d_\mathcal{E}}} \right),\\
    &=2-2\Tr\left(\sqrt{\rho}\sqrt{\Phi_A(\rho)}\otimes\frac{\ket{e_0}\bra{e_0}}{\sqrt{d_\mathcal{E}}}\right),\\
    &=2-\frac{2}{\sqrt{d_\mathcal{E}}}\Tr\left(\sqrt{\rho}\sqrt{\Phi_A(\rho)}\right)\Tr(\ket{e_0}\bra{e_0}),\\
    &=2-\frac{2}{\sqrt{d_\mathcal{E}}}\Tr\left(\sqrt{\rho}\sqrt{\Phi_A(\rho)}\right).
\end{align}
\begin{align}
    \Delta I_\mathcal{E|S}^\He(\Omega_t,\Omega_0)&=\frac{1}{\sqrt{d_\mathcal{E}}}\left[2-2\Tr\left(\sqrt{\rho}\sqrt{\Phi_A(\rho)} \right) \right],\\
    &=\frac{1}{\sqrt{d_\mathcal{E}}}d^2_\He(\rho,\Phi_A(\rho)).
\end{align}

\end{document}